\documentclass[11pt]{article}

\usepackage{graphicx}
\usepackage{amsmath,amssymb}
\usepackage{amsthm}
\usepackage{epsfig}
\usepackage{fullpage}
\usepackage{times}

\newcommand{\myignore}[1]{}

\newenvironment{description*}%
  {\vspace{-1ex}\begin{description}%
    \setlength{\itemsep}{-0.5ex}%
    \setlength{\parsep}{0pt}}%
  {\end{description}}
\newenvironment{itemize*}%
  {\vspace{-1ex}\begin{itemize}%
    \setlength{\itemsep}{-0.5ex}%
    \setlength{\parsep}{0pt}}%
  {\end{itemize}}
\newenvironment{enumerate*}%
  {\vspace{-1ex}\begin{enumerate}%
    \setlength{\itemsep}{-0.5ex}%
    \setlength{\parsep}{0pt}}%
  {\end{enumerate}}

{\makeatletter
 \gdef\xxxmark{%
   \expandafter\ifx\csname @mpargs\endcsname\relax % in minipage?
     \expandafter\ifx\csname @captype\endcsname\relax % in figure/caption?
       \marginpar{xxx}% not in a caption or minipage, can use marginpar
     \else
       xxx % notice trailing space
     \fi
   \else
     xxx % notice trailing space
   \fi}
 \gdef\xxx{\@ifnextchar[\xxx@lab\xxx@nolab}
 \long\gdef\xxx@lab[#1]#2{{\bf [\xxxmark #2 ---{\sc #1}]}}
 \long\gdef\xxx@nolab#1{{\bf [\xxxmark #1]}}
}

\newtheorem{theorem}{Theorem}
\newtheorem{lemma}[theorem]{Lemma}
\newtheorem{corollary}[theorem]{Corollary}
\newtheorem{proposition}[theorem]{Proposition}

\newcommand{\Patrascu}{P\v{a}tra\c{s}cu}
\newcommand{\ceiling}[1]{\lceil #1\rceil}

\newcommand{\eps}{\varepsilon}

\newcommand{\calD}{\mathcal{D}}
\newcommand{\calX}{\mathcal{X}}
\newcommand{\calY}{\mathcal{Y}}
\newcommand{\calZ}{\mathcal{Z}}

\newcommand{\cl}[1]{\ensuremath {\sf #1}}

\newcommand{\E}{\mathbb{E}}

\newcommand{\ignore}[1]{}

\title{Finding the Median (Obliviously) with Bounded Space}
\author{
     Paul Beame\thanks{University of Washington.  Research supported by NSF grants CCF-1217099 and CCF-0916400}
\and Vincent Liew\thanks{University of Washington.  Research supported by NSF grant CCF-1217099}
\and Mihai P\v{a}tra\c{s}cu\thanks{Much of this work was done with Mihai in
2009 and 2010 when the lower bounds for oblivious algorithms were obtained.  
This paper is dedicated to his memory.}
}

\begin{document}

\maketitle

\begin{abstract}
We prove that any oblivious algorithm using space $S$ to find the
median of a list of $n$ integers from $\{1,\ldots,2n\}$ requires time
$\Omega(n\log\log_S n)$.  This bound also applies to the problem of
determining whether the median is odd or even.
It is nearly optimal since Chan, following Munro and Raman, has shown
that there is a (randomized) selection algorithm using only $s$ registers,
each of which can store an input value or $O(\log n)$-bit counter,
that makes only $O(\log\log_s n)$ passes over the input.
The bound also implies a size lower bound for read-once branching
programs computing the low order bit of the median and
implies the analog of $\cl{P}\ne \cl{NP}\cap\cl{coNP}$ for 
length $o(n\log\log n)$ oblivious branching programs.
\end{abstract}

\section{Introduction}

The problem of selection or, more specifically, finding the median of a
list of values is one of the most basic computational problems.
Indeed, 
the classic deterministic linear-time median-finding algorithm of \cite{bfprt:median},
as well as the more practical expected linear-time randomized algorithm
QuickSelect are among the most widely taught algorithms.

Though these algorithms are asymptotically optimal with respect to time, 
they require substantial manipulation and re-ordering of the input during
their execution. 
Hence, they require the ability to write into a linear number of memory
cells. (These algorithms can be implemented with only $O(1)$ memory locations
in addition to the input if they are allowed to overwrite the input memory.)
In many situations, however, the input is stored separately and 
cannot be overwritten unless it is brought into working
memory.
The number of bits $S$ of working memory that an algorithm with read-only
input uses is its {\em space}.
This naturally leads to the question of the tradeoffs between the time $T$
and space $S$ required to find the median, or for selection more generally.  

Munro and Paterson~\cite{mp:selectsort-journal} gave multipass algorithms that
yield deterministic time-space tradeoff upper bounds for selection for small
space algorithms and showed that the number of passes $p$ must be
$\Omega(\log_s n)$ where $S=s\log_2 n$.
Building on this work,
Frederickson~\cite{frederickson:sortselect-journal} extended the range of
space bounds
to nearly linear space, deriving a multipass algorithm achieving
a time-space tradeoff of the form $T=O(n\log^* n+n\log_s n)$.
In the case of randomly ordered inputs, Munro and Raman~\cite{mr:selectrom-journal} showed that 
on average an even better upper bound of $p=O(\log\log_s n)$ passes and
hence $T=O(n\log\log_s n)$ is possible.
Chakrabarti, Jayram, and \Patrascu~\cite{cjp:median} showed that this is
asymptotically
optimal for multipass computations on randomly ordered input streams. 
Their analysis also applied to algorithms that perform arbitrary operations
during their execution.  

Chan~\cite{chan:selection-journal} showed how to extend the ideas of Munro
and Raman~\cite{mr:selectrom-journal} to
yield a randomized median-finding algorithm achieving the same time-space
tradeoff upper bound as in the average case that they analyze.  The resulting
algorithm, like all of those discussed so far, only accesses its input using
comparisons.    Chan coupled this algorithm with a corresponding time-space
tradeoff lower bound of $T=\Omega(n\log\log_S n)$ for 
randomized comparison branching programs, which implies the same
lower bound for the randomized comparison RAM model.  This is the first lower
bound for selection allowing more than multipass access to the input;
the input access can be input-dependent but the algorithm must
base all its decisions on the input order.   
Though a small gap remains because $S\ne s$,
the main question left open by~\cite{chan:selection-journal} is that of
finding time-space tradeoff lower bounds for median-finding algorithms that
are not restricted to the use of comparisons.   

\vskip 1ex
\begin{sloppypar}
\noindent{\it Comparison-based versus general algorithms\ }
Though comparison-based algorithms for selection may be natural, when the input
consists of an array of $O(\log n)$-bit integers, as one often assumes, 
there are natural alternatives to comparisons such as hashing that might
potentially yield more efficient algorithms.   Though comparison-based
algorithms match the known time-space tradeoff lower bounds in efficiency for
sorting when time $T$ is
$\Omega(n\log n)$~\cite{bc:sorting,bea:sorting,pr:comparison-sorting},
they are powerless in the regime when $T$ is $o(n\log n)$.
Moreover, if one considers the closely related problem of element distinctness,
determining whether or not the input has duplicates, the known time-space
tradeoff lower bound of $T=\Omega(n^{2-o(1)}/S)$ for (randomized) comparison
branching programs~\cite{yao:ED-journal} can be beaten for 
$S$ up to $n^{1-o(1)}$ by an algorithm using hashing~\cite{bcm:windowsED} that
achieves $T=\tilde O(n^{3/2}/S^{1/2})$\footnote{We use $\tilde O$ and $\tilde \Omega$ notations to hide logarithmic factors.}.
Therefore, the restriction to comparison-based algorithms can be a significant
limitation on efficiency.
\end{sloppypar}

\vskip 1ex
\begin{sloppypar}
\noindent{\it Our results\ }
We prove
a tight $T=\Omega(n\log\log_S n)$ lower bound for median-finding using
arbitrary oblivious algorithms.  
Oblivious algorithms are those that can access the data in any order, not just
in a fixed number of sweeps across the input, but that order cannot be data
dependent.
Our lower bound applies even for the decision problem of computing
{\sc MedianBit}, the low order bit of the median,
when the input consists of $n$ integers chosen from $\{1,\ldots, 2n\}$.
This bound substantially generalizes the lower bound of~\cite{cjp:median} for
multipass median-finding algorithms.  
Though our lower bound does not apply when there is input-dependent access to
the input, it allows one to hash the input data values into working storage,
and to organize and manipulate working storage in arbitrary ways.
\end{sloppypar}

The median can be computed by a simple nondeterministic
oblivious read-once branching program of polynomial size that guesses
and verifies which input integer is the median. 
When expressed in terms of size for time-bounded oblivious branching
programs our lower bound therefore shows that for every time bound $T$ that is
$o(n\log\log n)$,
{\sc MedianBit} and its complement have nondeterministic oblivious branching
programs of polynomial size but {\sc MedianBit} requires super-polynomial size
deterministic oblivious branching programs, hence separating
the analogs of \cl{P} from $\cl{NP}\cap\cl{coNP}$.

We derive our lower bound using a reduction from a new communication 
complexity lower bound for two players to find
the low order bit of median of their joint set of input integers in a bounded
number of rounds.
The use of communication complexity lower bounds in the ``best partition" model
to derive lower bounds for oblivious algorithms is not new, but the necessity
of bounded rounds is.
We derive our bound via a round-preserving reduction
from oblivious computation to best-partition communication
complexity~\cite{okol93,ajtai:nondetBP-journal}.  This reduction is
asymptotically less efficient than the reductions of~\cite{am:meanders,brs93}
but the latter do not preserve the number of rounds, which is essential here
since there is a very efficient $O(\log n)$-bit communication protocol using
an unbounded number of rounds~\cite{kn97}.
Moreover, the loss in efficiency does not prevent us from achieving
asymptotically optimal lower bounds.
%There are several reductions that can be used to convert oblivious computation
%to best-partition communication complexity:
%greedy methods~\cite{am:meanders,bns89,bv:multibp},
%random methods~\cite{brs93,bjs:tradeoff,bssv:randomts-journal,bv:multibp},
%and pigeonhole principle-based methods~\cite{okol93,ajtai:nondetBP-journal}.
%Of these, the random methods in general yield the largest lower bounds, but
%they do not preserve the number of rounds.
%In the case of the median problem, round preservation is essential since there
%is a very efficient $O(\log n)$-bit communication protocol using an unbounded
%number of rounds~\cite{kn97}.
%The greedy methods are awkward to apply and are not clearly applicable here,
%so we derive a our bound via a round-preserving reduction
%from~\cite{okol93,ajtai:nondetBP-journal}. 
%Though the reduction is asymptotically less efficient than the randomized
%methods in general, it is efficient enough in the context of our
%median-finding lower bound to yield asymptotically optimal bounds.

We further show that the fact that the median function is symmetric in its
inputs implies that our oblivious branching program lower bound also applies
to the case of non-oblivious read-once branching programs.
Ideally, we would like to extend our non-oblivious results to larger time
bounds.
However, we show that extending our lower bound even to read-twice branching
programs in the non-oblivious case would require fundamentally new lower
bound techniques.
The hardness of the median problem is essentially that of a decision
problem:
Though the median problem has $\Theta(\log n)$ bits of output, the high order
bits of the median are very easy to compute; it is really the low order
bit, {\sc MedianBit}, that is the hardest to produce and encapsulates all of
the difficulty of the problem. 
Moreover, all current methods for time-space tradeoff lower bounds for
decision problems on general branching programs, and indeed for read-$k$
branching programs for $k>1$, also apply to nondeterministic algorithms
computing either the function or its complement and hence cannot apply to the
median because it is easy for such algorithms.

\section{Preliminaries}

Let $D$ and $R$ be finite sets.  We first define branching programs that compute
functions $f:D^n\rightarrow R$:
A \textit{$D$-way branching program} is a
connected directed acyclic multigraph with special nodes: the
\textit{source node} and possibly many \textit{sink nodes}, a sequence of
$n$ input values and one output.
Each non-sink node is labeled with an input index and every edge is
labeled with a symbol from $D$, which corresponds to the value of
the input indexed at the originating node; there is precisely one
out-edge from each non-sink node labeled by each element of $D$.  
We assume that each sink node is labeled by an element of $R$.
The time $T$ required by a branching program is the length of the longest
path from the source to a sink and the space $S$ is $\log_2$ of the number
of nodes in the branching program.
A branching program is {\em leveled} iff all the paths from the source to any
given node in the program are of the same length; a branching program can be
leveled by adding at most $\log_2 T$ to its space.

A branching program $B$ computes a function $f_B:D^n\rightarrow R$ by
starting at the source and then proceeding along the nodes of the
graph by querying the input locations associated with each node and following
the corresponding edges until it reaches a sink node; the label of the sink
node is the output of the function.

A branching program is {\em oblivious} iff on every path from the 
source node to a sink node, the sequence of input indices is precisely the same.
It is (syntactic) {\em read-$k$} iff no input index appears more than $k$
times on any path from the source to a sink. 

\ignore{% We don't get randomized lower bounds yet
A {\em randomized $D$-way branching program} is a probability distribution
$\mathcal B$ over (deterministic) $D$-way branching programs.   It is said
to compute $f$ with probability at most $\varepsilon$-error iff for every 
$x\in D^n$, the probability over $B\sim \mathcal{B}$ that $B(x)=f(x)$
is at least  $1-\varepsilon$.
%It is oblivious (respectively, read-$k$) iff each the branching programs in
%its support is\footnote{The definition we use is that of randomized oblivious
%algorithms, in constrast to oblivious randomized algorithms in which the
%distribution of accesses over all choices of the randomness is independent
%of the input but the order of access for each fixed choice of the random bits
%may depend on the input.  See~\cite{bm:robp} for a discussion of the
%differences.}

Given a probability distribution $\mu$ on $D^n$, a deterministic branching
program $B$ computes $f$ with (distributional) error at most $\varepsilon$ iff
the probability over $x\sim \mu$ that $f_B(x)=f(x)$ is at least $1-\varepsilon$.
Yao's lemma implies that lower bounds for the distributional error for fixed
time and space parameters imply lower bounds for these on randomized branching
programs with the same error.
}

Branching programs can easily simulate any sequential model of computation
using the same time and space bounds.   In particular branching programs
using time $T$ and space $S$ can simulate random-access machine (RAM) 
algorithms using time $T$ measured in the number of input locations queried
and space $S$ measured in the number of bits of read/write storage required.  
The same applies to the simulation of randomized RAM algorithms by
randomized branching programs.

We also find it useful to discuss nondeterministic branching programs for
(non-Boolean) functions, which simulate nondeterministic RAM algorithms for 
function computation.  These have the property that
multiple outedges from a single node can have the same label and outedges for
some labels may not be present.   
Every input must have at least one path that leads to a sink and
all paths followed by an input vector that lead
to a sink must lead to the same one, whose label is the output value of the
program. 
This is different from the usual version for decision problems in which one only
considers accepting paths and infers the output value for those that are not
accepting.  When we consider Boolean functions we will typically assume the
usual version based on accepting paths only.

We consider bounded-round versions of deterministic and randomized
two-party communication
complexity in which two players Alice and Bob receive $x\in \calX$ and
$y\in \calY$ and cooperate to
compute a function $f:\calX\times\calY\rightarrow \calZ$.  
A round in a protocol is a maximal segment of communication in which the
player who speaks does not change.
For a distribution $\calD$ on $\calX\times\calY$, we say that a 2-party
deterministic communication protocol computes $f$ with error at most
$\varepsilon<1/2$ under $\calD$ iff the probability over $\calD$ that the
output of the
protocol on input $(x,y)\sim \calD$ is equal to $f(x,y)$ is at least
$1-\varepsilon$.
As usual, via Yao's lemma, for any such distribution $\calD$, the minimum
number of bits communicated by any deterministic protocol that computes $f$
with error at most $\varepsilon$ is a lower bound on the number of bits
communicated by any (public coin) randomized protocol that computes $f$ with
error at most $\varepsilon$.

We say that a 2-party deterministic communication protocol has parameters
$[P, \eps; m_1, m_2, \dots]$ for $f$ over a distribution $\calD$ if:
\begin{itemize*}
\item the first player to speak is $P \in \{ A, B \}$;
\item it has error $\eps < \frac{1}{2}$ under input distribution $\calD$;
\item the players alternate turns, sending messages of $m_1, m_2, \dots$ bits,
respectively.
\end{itemize*}

For probability distributions $P$ and $Q$ on a domain
$U$, the statistical distance between $P$ and $Q$, is
$||P-Q||=\max_{A\subseteq U} |P(A)-Q(A)|$, which is 1/2 of the $L_1$ distance
between $P$ and $Q$.  
Let $\log$ denote $\log_2$ unless otherwise specified. i
Let $H(X)$ be the binary entropy of random variable $X$,
$H(X|Y)=\E_{y\sim Y} H(X|Y{=}y)$,
and let $I(X;Y|Z)$ be the mutual information between random variables $X$ and
$Y$ conditioned on random variable $Z$.
We have $I(X;Y|Z)\le H(X|Z)\le H(X)$.

\section{Round Elimination}

Let $f : \calX \times \calY \to \{0,1\}$ and consider a distribution
$\calD$ on $\calX \times \calY$. 
We define the 2-player communication problem $f^{[k]}$ as
follows: Alice receives $x \in \calX^k$, while Bob receives $y \in
\calY^k$ and $j\in [k]$; together they want to find $f(x_j, y_j)$.  
Also, given $\calD$ we define an input distribution $\calD^{[k]}$ for
$f^{[k]}$ by choosing each $(x_i, y_i)$ pair 
independently from $\calD$, and independently choosing $j$ uniformly from
$[k]$.

The following lemma is a variant of standard techniques and was suggested
to us by Anup Rao; its proof is in the appendix for completeness.

\begin{sloppypar}
\begin{lemma}    \label{lem:rel}
Assume that there exists a 2-party deterministic protocol for $f^{[k]}$ with
parameters $[A,\eps; m_1, m_2, m_3, \dots \big]$ over
$\calD^{[k]}$ where $m_1 = \delta^2 k/(8\ln 2)$.
Then there exists a 2-party deterministic protocol for $f$ with parameters
$\big[B, \eps + \delta; m_2, m_3, \dots \big]$ over $\calD$.
\end{lemma}
\end{sloppypar}

The intuition for this lemma is that, since $f^{[k]}$ has $k$ independent copies
of the function $f$ and Alice's first message has length at most $m_1$ which
is only a small fraction of $k$, there must be some copy of $f$ on which $B$
learns very little information.  This is so much less than one bit that $B$
could forego this information in computing $f$ and still only lose $\delta$
in his probability of correctness.
The quadratic difference between the number of bits of information per copy,
$\delta^2/(8\ln 2)$, and the probability difference, $\delta$, comes from
Pinsker's inequality which relates information and statistical distance.

\section{The Bounded-Round Communication Complexity of\\ (the Least-Significant Bit of) the Median}

We consider the complexity of the following communication
game.  Given a set $A$ of $n$ elements from $[2n]$ partitioned equally between
Alice and Bob, determine the least significant bit of the median of $A$.
(Since $n$ must be even in order for $A$ to be partitioned evenly, we take the
median to be $n/2$-th largest element of $A$.)
We consider the number of rounds of communication required when the length
of each message is at most $m$ for any $m\ge \log n$.

\subsubsection*{A Hard Distribution on Median Instances}

For our hard instances we first define a pairing of the elements of $[2n]$ that
depends on the value of $m$.  The set $A$ will include precisely one
element from each pair.
For the input to the communication problem, we randomly partition the pairs
equally between the two players which will therefore also automatically equally
partition the set $A$.
We then show how to randomly choose one element from each pair to include
in $A$.  

In the construction, we define the pairing of $[2n]$ recursively; the
parameters of each recursive pairing will depend on the initial value $n_0$ of
$n$.
Let $k=k(m,n_0) =  m\log^2 n_0$.
If $\sqrt{n}<k\log^3 n_0$ then the elements of $[1,2n]$ are simply paired
consecutively.
If $\sqrt{n}\ge k\log^3 n_0$ then the pairing of $[2n]$ consists of a ``core''
of $\gamma = \sqrt{n} / \log^2 n_0$ pairs, plus $n-\gamma$ ``shell'' pairs
on $[1,n-\gamma]\cup [n+1+\gamma,2n]$.  
In the shell, $i$ and $2n+1-i$ are paired.
The core pairs are obtained by embedding $k$ recursive instances (using 
the same values of $m$ and $n_0$) of
$n'=\frac{\gamma}{k}$ pairs each on consecutive sets
of $\frac{2\gamma}{k}$ elements, and placing them back-to-back in the value
range $[n-\gamma + 1, n+\gamma]$, see Figure~\ref{fig:shell-core}.
The size of the problem at each level of recursion decreases from
$n$ to $n'=\gamma/k = \sqrt{n}/ (m\log^4 n_0)$.
\begin{figure}[t]
\centering
\includegraphics[scale = 0.25]{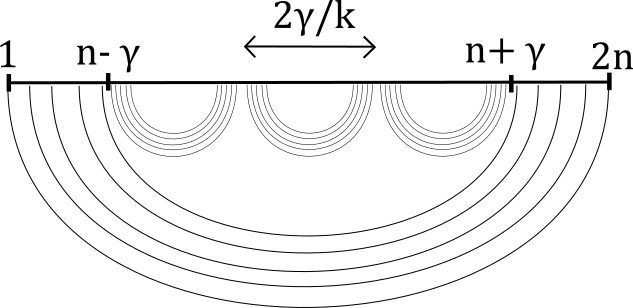}
\caption{\small Recursive construction of the pairing for the hard instances.}
\label{fig:shell-core}
\end{figure}
In determining the median, the only relevant information about the
shell elements is how many are below $n$; let this number be
$\frac{n}{2} - x$. If $x \in [1, \gamma]$, the median of the
entire array $A$ will be the $x$-th order statistic of the core.

If furthermore, $x = \frac{\gamma}{k} ( j - \frac{1}{2} )$ for an
integer $j$, the median of $A$ will be exactly the median of the
$j$-th embedded subproblem.  In our distribution of hard instances, we will
ensure that $x$ has this nice form.

Formally, the distribution $\calD^n_{m,n_0}$ of the hard instances $A$ of size
$n$ on $[2n]$ is the following.
Generate $k$ recursive instances on $\calD^{\gamma/k}_{m,n_0}$ and
place shifted versions of them back-to-back inside the core.
Choose $j \in [k]$ uniformly at random.
Choose $\frac{n}{2} - \frac{\gamma}{k} (j - \frac{1}{2})$ uniformly random
shell elements in $[1,n-\gamma]$ to include in $A$; 
for every $i \in [1,n-\gamma]\setminus A$, we have
$2n+1-i\in A$.
This will ensure that the median of $A$ is precisely the median of the
$j$-th recursive instance inside the core.

\begin{sloppypar}
Initially we have $n=n_0$ and the recursion only continues when
$\gamma=\sqrt{n}/\log^2 n_0\ge k\log n_0$, so in the base case we have
at least $\log n_0$ elements.
In this case, the $i$-th element is chosen randomly and uniformly
from the paired elements $2i-1$ and $2i$ and so
the least significant bit of the median is uniformly chosen
in $\{0,1\}$.
\end{sloppypar}

\begin{sloppypar}
The size of the problem after $t$ levels of recursion
remains at least $n_0^{1/2^t}/(m\log^4 n_0)^{2-1/2^{t-1}}$ and our definition
gives at least $t$ levels provided that this size
$n_0^{1/2^t}/(m\log^4 n_0)^{2-1/2^{t-1}}\ge \log n_0$; i.e., 
$n_0\ge m^{2^{t+1}-2} \log^{9\cdot 2^t-2}n_0$.
We will show that after one message for each level of recursion, the answer
is still not determined.  
\end{sloppypar}

%To summarize the structure of the hard instance, we can say that each
%level selects a uniformly random child among $k$ possibilities. In the
%base case, a single bit (the low order bit of the median) gives the
%answer. 
The general idea of the lower bound is that each round of communication, which
consists of at most $m$ bits and is much smaller than the branching factor $k$,
will give almost no information about a typical recursive subproblem in the
core.   

We use the round elimination lemma to make this precise, and with it derive
the following theorem:
\begin{sloppypar}
\begin{theorem}
\label{thm:cc}
If, for $A$ chosen according to $\calD^n_{m,n}$ and partitioned randomly, Alice and Bob determine the least
significant bit of the median of $A$ with bounded error $\varepsilon<1/2$ using
$t$ messages of at most $m \ge \log n$ bits each, then 
$m^{2^{t+1}-2}>n/ \log^{9\cdot 2^t-2}n$, which implies that
$t\ge \log\log_m n - c$ for some constant $c$.
\end{theorem}
\end{sloppypar}

\subsubsection*{The Partition Between the Players}

To ensure that neither player has enough information to skip a level of
the recursion, we insist that the shell for each subproblem be nicely
partitioned between the two players. 
For any given shell there is a set of $n'>m^2/2\ge  0.5\log^2 n_0$ shell pairs.
Since a player receives a random 1/2 of all pairs, by Hoeffding's inequality,
with probability $2^{-\Omega(n')}$, which is $n_0^{-\Omega(\log n_0)}$,
at least $\frac{n'}{3}$ pairs go to each player.
We can use this to say that with high probability at least 1/3 of all shell
elements at a level go to each player at every level of the recursion:
This follows easily because over all levels of the recursive pairing,
there are only a total of $o(\sqrt{n_0})$ different shells associated with
subproblems and each one fails only with probability $n_0^{-\Omega(\log n_0)}$.

From now on, fix a partition satisfying the above requirement at all
recursion nodes. We will prove a lower bound for any partition
satisfying this property. Since we are discarding $o(1)$ of possible
partitions, the error of the protocol may increase by $o(1)$, which is
negligible.

\subsubsection*{The Induction}

Our proof of Theorem~\ref{thm:cc} will work by induction, using the following
message elimination lemma:

\begin{lemma}
\label{lem:median-elim}
Assume that there is a protocol for the median on instances of size $n$,
with error $\eps$ on $\calD^n_{m,n_0}$ for $\sqrt{n}\ge k\log n_0=m\log^3 n_0$, using $t$ messages of size at most
$m$ starting with Alice.
Then, there is a
protocol for a subproblem of size $\gamma/k$, with error $\eps +
O(\frac{1}{\log n_0})$ on $\calD^{\gamma/k}_{m,n_0}$, using $t-1$ messages of
size at most $m$ starting with Bob.
\end{lemma}

We use Lemma~\ref{lem:median-elim} to prove Theorem~\ref{thm:cc} by inductively
eliminating all messages.
Let $n_0=n$.  
At each application we remove one message to get an error increase
of $O(\frac{1}{\log n_0})$.
If the number of rounds is less than the number of levels of recursion, i.e., $m^{2^{t+1}-2}\le n/ \log^{9\cdot 2^t-2}n$, then the {\sc MedianBit} value
of the subproblem will still be a uniformly random bit on the remaining input,
but the protocol will have no communication and the error will have increased
to at most $\epsilon+O(\frac{t}{\log n})<1/2$
since $t$ is $O(\log\log_m n)$, which is a contradiction.

To prove Lemma~\ref{lem:median-elim} we want to apply Lemma~\ref{lem:rel} 
using the $k$ subproblems in the core, but the assumption of Lemma~\ref{lem:rel}
requires that (1) Alice does not know anything about which subproblem $j\in [k]$
is chosen by Bob, and (2) that subproblem $j$ is chosen uniformly at random.  
The choice of subproblem $j$ is determined by the shell elements at this level.

Denote Alice's shell elements by $x^s$, and Bob's shell elements by
$y^s$. Let Alice's part of the core subproblems be $x_1, \dots,
x_k$, and Bob's part be $y_1, \dots, y_k$. Note that the choice of the relevant
subproblem $j$ is some function of $(x^s, y^s)$, and the median of
the whole array is the median of $x_j \cup y_j$.

The proof of Lemma~\ref{lem:median-elim} proceeds in two stages:

\paragraph{Fixing $x^s$.}
We first fix the value of $x^s$ so that the choice of subproblem does not depend
on Alice's input and, moreover, so that the probabilities for different
values of $j$ over Bob's input $y^s$ will not be very different from each
other because they are still near the middle binomial coefficients.   

By the niceness of the partition of the pairs, we know that the number of
Alice's shell pairs is
$|x^s| \in \big[ \frac{1}{3} (n - \gamma), \frac{2}{3} (n - \gamma) \big]$.
Let $a$ be the number of elements in $x^s$ that are below $n$. 
We want to fix $x^s$ such that the error
does not increase too much, and
$|a -\frac{|x^s|}{2}| \le \sqrt{n} \cdot \log n_0$:

%The probability that any element is below $n$ is $\frac{1}{2} \pm
%O(\frac{\gamma}{n}) = \frac{1}{2} \pm O(\frac{1}{\sqrt{n} \log^2 n})$,
%depending on the value $j$ in our instance. By Chernoff, $a =
%\frac{|x^s|}{2} \pm O(\sqrt{n} \cdot \log n)$ with high probability.
%\footnote{We actually need a standard trick, since we cannot apply
  %Chernoff directly. Indeed, the variables are not independent: they
  %have a sum prescibed by $j$. However, we can pretend they're
  %independent, and obtain a high probability bound. The probability
  %that the sum is correct, given $j$, is $\Omega(\frac{1}{\sqrt{n}})$;
  %thus, we maintain the high probability guarantee even if we
  %condition on this event, returning to the true probability space.}

No matter which value of $j\in [k]$ is chosen in the input distribution,
the shell elements chosen to be below $n$ consist of a random subset of
$x^s\cup y^s$ of a fixed size that is between $n/2-\gamma$ and $n/2+\gamma$; i.e.,
of fractional size $p_j$ between $\frac12-\frac{\gamma}{n}$ and $\frac12+\frac\gamma{n}$.
By Hoeffding's inequality,
the probability that the actual number $a$ of these elements that
land in $x^s$ deviates from $|x^s|/2$ by more than $(t+\frac{\gamma}{n})|x^s|$
is at most $2e^{-2t^2|x^s|}$.  
Since $(n-\gamma)/3\le |x^s| \le 2(n-\gamma)/3$, the probability that 
this deviates from $|x^s|/2$ by more than $\sqrt{n}\log n_0$ is
at most $n_0^{-O(\log n_0)}$.
We discard all values of $x^s$ that lead to $a$ outside this range. 
Now fix $x^s$ to be the value that minimizes the conditional error.

\paragraph{Making $j$ uniform.}
Once $x^s$ is fixed, $j$ is a function only of $y^s$. Thus, we are close to the
setup of Lemma~\ref{lem:rel}: Alice
receives $x_1, \dots, x_k$, Bob receives $y_1, \dots, y_k$ and $j\in
[k]$, and they want to compute a function $f(x_j, y_j)$.
The only problem is that
the lemma requires a uniform distribution of $j$, whereas our
distribution is no longer uniform (having fixed $x^s$).
However, we will argue that it is not far from uniform.

%For this, we use:
%
%\begin{lemma}
%Let $\Delta \in [-\frac{n}{10}, \frac{n}{10}]$, and $0 < \delta <
%|\Delta|$ such that $\Delta \delta < n$. Then $\binom{n}{n/2 + \Delta
  %+ \delta} = \binom{n}{n/2 + \Delta} \cdot \big( 1 \pm
%\frac{O(\delta\Delta)}{n} \big)$.
%\end{lemma}
%
%\begin{proof}
%This is a simple calculation. If $\Delta < 0$, use the symmetry
%$\binom{n}{n/2 - \Delta} = \binom{n}{n/2 + \Delta}$. If $\Delta > 0$,
%adding $\delta$ will increase the binomial coefficient by $\delta$
%factors of the form $\frac{n/2 + \Delta + i + 1}{n/2 - \Delta - i} \le
%1 + \frac{2\Delta + i}{n/2 - \Delta - i} = 1 +
%O(\frac{\Delta}{n})$. The product of these factors is bounded by
%$\big( 1 + \frac{O(\Delta)}{n} \big)^\delta \le 1 + \frac{O(\Delta
  %\delta)}{n}$.
%\end{proof}

For each fixed $j_0\in [k]$, 
%what is $\Pr[j = j_0]$?  If 
if $a$ shell elements from Alice's part are below $n$, then
Bob must have $\frac{n}{2} - a - \frac{\gamma}{k} ( j_0 - \frac{1}{2})$
shell elements below $n$.
Therefore, $\Pr[j = j_0]$ is proportional to
$\displaystyle \binom{ |y^s| }{ \frac{n}{2} 
      - a - \frac{\gamma}{k} ( j_0 - \frac{1}{2}) }.$
More precisely $\Pr[j = j_0]$ is this binomial coefficient
divided by the sum of the coefficients for all $j_0$. Thus, to
understand how close $j$ is to uniform, we must understand the
the dependence of these binomial coefficients on $j_0$.

Let $\Delta=a-|x^s|/2$.  This satisfies $|\Delta|\le \sqrt{n}\log n_0$.
Since $|y^s| = n - |x^s| \ge \frac{n-\gamma}{3}> n/4$
we have 
$\binom{ |y^s| }{ \frac{n}{2} - a - \frac{\gamma}{k} ( j_0 - \frac{1}{2}) } 
=\binom{|y^s|}{|y^s|/2 -\Delta - \delta_{j_0}}$
where $0<\delta_{j_0}<\gamma$.
Assume wlog that $\Delta\ge 0$.
The ratio between different binomial coefficients is at most
the ratio 
\begin{align*}
\binom{n/4}{n/8-\Delta}/\binom{n/4}{n/8-\Delta-\gamma}
&=\frac{(n/8+\Delta+\gamma)\cdots (n/8+\Delta+1)}{(n/8-\Delta)\cdots (n/8-\Delta-\gamma+1)}\\
&\le \left(1+ \frac{10(2\Delta+\gamma)}{n}\right)^\gamma
\end{align*}
which is $1+O(\frac{\Delta\gamma}{n})=1+O(\frac{1}{\log n_0})$ given the
values of $\Delta$ and $\gamma$.

Therefore we have shown that the statistical distance between the induced
distribution on $j$ and the uniform distribution is $O(\frac{1}{\log n_0})$.
We can thus consider the following alternative distribution for the problem:
pick $j$ uniformly at random, and manufacture $y^s$ conditioned on
this $j$. The error on the new distribution increases by at most
$O(\frac{1}{\log n_0})$.

Now we can apply the round elimination lemma, Lemma~\ref{lem:rel}. 
As $k \ge m \log^2 n_0$, the lemma will increase the error by
$O(\frac{1}{\log n_0})$.

\section{Oblivious Branching Programs and the Median}

The following result is essentially due to Okol'nishnikova~\cite{okol93}, who
used it with slightly different parameters for read-$k$ branching programs, and
was independently derived by Ajtai~\cite{ajtai:nondetBP-journal} in the context
of general branching programs.

\begin{proposition}
\label{prop:assign-layers}
Let $s$ be a sequence of of $kn$ elements from $[n]$.  If $s$ is divided into
$r=4k^2$ segments $s_1,\ldots, s_{r}$, each of length $n/(4k)$, then
there is an assignment of $2k$ segments $s_j$ to a set $L_A$ and all remaining
segments $s_j$ to $L_B$ so that the number
$n_A$ ($n_B$) of elements of $[n]$
whose only appearances are in segments in $L_A$ (respectively, $L_B$) satisfy
$n_A\ge n/(2\binom{4k^2}{2k})$ and $n_B\ge n/2$.
\end{proposition}

\begin{proof}
There is a subset $V$ of at least $n/2$ elements of $[n]$ that occur at most
$2k$ times in $s$ and hence appear in at most $2k$ segments of $s$.
Choose the $2k$ sets $s_j$ to include in $L_A$ uniformly at random.
For a given $i\in V$, $i$ will contribute to $n_A$ if and only if all of the 
the at most $2k$ segments that contain its
occurrences are chosen for $L_A$.  This occurs with probability at least
$1/\binom{r}{2k}$; hence the expected number of elements in $V$ that only occur
in segments of $L_A$ is at least $|V|/\binom{r}{2k}$. 
Therefore we can select a fixed assignment that contains has at least this
number.
Since the total length of segments in $L_A$ is at most $2k n/(4k)\le n/2$, at
least $n/2$ elements of $[n]$ only occur in segments in $L_B$.
\end{proof}

\begin{lemma}
\label{lem:oblivious-reduce}
Suppose that there is a $2n$-way oblivious branching program of size $2^S$
running in time
$T=kn$ that computes {\sc MedianBit} for $n$ distinct inputs
from $[2n]$.   Then there is deterministic 2-party communication protocol 
using at most $4k$ messages of $S$ bits each plus a final 1-bit message to
compute {\sc MedianBit}
for $N=\ceiling{n/\binom{4k^2}{2k}}$ distinct inputs from $[2N]$ that
are divided evenly between the two players.
\end{lemma}

\begin{proof}
Let $s$ be the length $T$ sequence of indices of inputs queried by the oblivious
branching program.
Let $k=T/n$, $r=4k^2$, and $N=\ceiling{n/\binom{r}{2k}}$.
Fix the assignment of segments to $L_A$ and $L_B$ given
by Proposition~\ref{prop:assign-layers}. 
Arbitrarily select a subset $I_A$ of $N/2$ of the $n_A$ indices that only appear
in $L_A$ and
give those inputs to player $A$.  Similarly, select a subset $I_B$ of $N/2$ of
the $n_B$ indices
that only appear in $L_B$ and give those inputs to player $B$.  
Let $Q$ be the remaining set of $n-N$ input indices.   

Fix any input assignment to the indices in $Q$ that assigns $(n-N)/2$
distinct values from $[n-N]$ to half the elements of $Q$ and
the same number of distinct values from $[n+N+1,2n]$ to the other half
of the elements of $Q$.   After fixing this partial assignment we restrict
the remaining inputs to have values in the segment $[n-N+1,n+N]$ of length
$2N$.  

The communication protocol is derived as follows: Alice (resp. Bob)
interprets her $N/2$ inputs from $[2N]$ as assignments from $[2n]$ to the
elements of $I_A$ (resp. $I_B$) by adding $n-N$ to each value.  Alice will
simulate the branching program executing the segments in $L_A$ and Bob
will simulate the branching program executing the segments in $L_B$.
A player will continue the simulation until the next segment is held by
the other player, at which point that
player communicates the name of the node in the branching program reached at
the end of its layer.  Since $L_A$ has only $2k$ segments, there are at most
$4k$ alternations between players as well as the final output bit which
gives the total communication.  By construction, the median of the whole
problem is the median of the $N$ elements and the final answer for
{\sc MedianBit} on $[2N]$ is computed by
XOR-ing the result with the low order bit of $n-N$.
\end{proof}

\begin{theorem}
Any oblivious branching program computing {\sc MedianBit}
for $n$ inputs from $[2n]$ in time $T\le kn$ requires size at least
$2^{\tilde \Omega(n^{1/2^{4k+2}})}$;
in particular, if it uses space $S$, any oblivious branching program requires
time $T\ge 0.25 n\log\log_S n - c\;n$ for some constant $c$.
\end{theorem}

\begin{proof}
Since $T/n\le k$, applying
Lemma~\ref{lem:oblivious-reduce} we derive a 2-party communication
protocol sending $t=4k+1$ messages of at most $S\ge \log n$ bits each
to compute {\sc MedianBit} on
$N\ge n/\binom{4k^2}{2k} \ge n/(2ek)^{2k}$ inputs from $[2N]$.
By Theorem~\ref{thm:cc}, $S> N^{1/(2^{t+1}-2)}/\log^{(9\cdot 2^t-2)/(2^{t+1}-2)}N
> N^{1/(2^{4k+2}-2)}/\log^{71/15} N$ since $t\ge 4$
and hence 
$S\ge n^{1/(2^{4k+2}-2)}/\log^5 n$.
The size of the branching program is $2^S$ where $S$ is its space.
Moreover, taking logarithms base $S$ and then base 2 we have
$4k\ge \log\log_S n - c'$ for some constant $c'$.
\end{proof}

\subsection*{Analog of $\cl{P} \ne \cl{NP}\cap\cl{coNP}$ for time-bounded oblivious BPs }

\begin{corollary}
\label{main-corollary}
Any oblivious branching program of length $T\le kn$ computing the low
order bit of the median requires size at least
$2^{\tilde\Omega(n^{1/2^{4k+2}})}$;
in particular, this size is super-polynomial when $T$ is $o(n\log\log n)$.
\end{corollary}

On the other hand, the median can be computed by a nondeterministic
oblivious read-once branching program using only $O(\log n)$ space.

\begin{lemma}
\label{upper}
There is a nondeterministic oblivious read-once branching program of size
$O(n^4)$ that computes the median on $n$ integers from $[2n]$.
\end{lemma}

\begin{proof}
The branching program guesses the value of the median in $[2n]$ 
and keeps track of the number of elements that it has seen
both less than the median and equal to the median in order to check that the
value is correct.
Other than the source and sink nodes there is one node of the branching
program for each $(i,m,\ell,e)$ for $m\in [2n]$, $i\in [n]$ such that
$0\le \ell+e\le \min(i,(n+1)/2+1)$.   The source node which queries $x_1$ is 
the only node to have multiple outedges with the same value label.  It has
$(2n)^2$ outedges, $2n$ for each value, one corresponding to each of the
median value guesses.  If at a node $(i,m,\ell,e)$, the values $i$, $\ell$, $e$
together with the value $j$ of $x_{i+1}$
are inconsistent with $m$ being the median then the outedge for $j$ is not
present.
\end{proof}

In particular, in contrast to Corollary~\ref{main-corollary}, Lemma~\ref{upper}
implies that {\sc MedianBit} can be computed in polynomial size
by length $n$ nondeterministic and co-nondeterministic oblivious branching
programs, hence we have shown the analog of $\cl{P}\ne \cl{NP}\cap\cl{coNP}$
for oblivious branching programs of length $o(n\log\log n)$.

\section{Beyond Oblivious Branching Programs}

%\xxx{Good question.}

We first observe that our lower bounds for the median problem extend to
the case of read-once branching programs by using the fact that such
programs for the median can also be assumed to be oblivious without loss
of generality.   (Oblivious read-once branching programs are also
known as \emph{ordered binary decision diagrams (OBDDs)}.)  

\begin{lemma}
If $f:D^n\rightarrow R$ is a symmetric function of its inputs then 
for every 
%(randomized) 
read-once branching $B$ computing $f$ 
%(with error $\varepsilon$ under a symmetric distribution $\mu$) 
there is an 
%(randomized)
oblivious read-once branching program,
of precisely the same size as $B$,
that computes $f$.
%(with error $\varepsilon$).
\end{lemma}

\begin{proof}
With each node $v$ in a read-once branching program, we can associate a
set $I_v\subseteq [n]$ of input indices that are read along paths from the
source node to $v$.   We make $B$ into an oblivious branching program
by replacing the index at node $v$ by $|I_v|+1$.   This yields an 
oblivious read-once branching program (not necessarily leveled) that reads its
inputs in the order $x_1,x_2,\ldots,x_n$ along every path (possibly skipping
over some inputs on the path).  
Since $f$ is a symmetric function, a path of length $t\le n$ in $B$
queries $t$ different input locations and the value of the function on the
partial inputs is the same because the function is symmetric and the values in
those $t$ input locations are the same.   
%Since $\mu$ is symmetric and no input
%is queried more than once, the probability that the
%path is followed is also identical in the two cases. 
\end{proof}

%Since the distribution for our median lower bound is symmetric 
We immediately obtain the following corollary.

\begin{corollary}
For any $\varepsilon<1/2$, any read-once branching program computing
{\sc MedianBit} for $n$ integers from $[2n]$ requires
size $2^{n^{\Omega(1)}}$.
\end{corollary}

In particular this means that {\sc MedianBit}
is another example, after those in~\cite{jrsw:pvnpconp}, of a problem showing 
the analogue of  $\cl{P}\ne \cl{NP}\cap \cl{coNP}$ for read-once branching
programs.
However, proving the analogous property even for read-twice branching programs
remains open and will require a fundamentally new technique for deriving
branching program lower bounds.

The approach in all lower bounds for general branching programs 
(or even for read-$k$ branching programs) computing decision
problems~\cite{brs93,okol93,bjs:tradeoff,ajtai:nondetBP-journal,ajtai:nonlinear,bssv:randomts-journal,bv:multibp} applies equally well to nondeterministic
computation.
(For example, the fact that the technique also works for nondeterministic
computation is made explicit in~\cite{brs93}.)
Though this technique has been used to separate nondeterministic from
deterministic computation~\cite{ajtai:nondetBP-journal} computing a Boolean
function $f$, it is achieved by proving a nondeterministic lower bound for 
computing $\overline{f}$.
Since the nondeterministic oblivious read-once branching program computing the
median has $T=n$ and $S=O(\log n)$, the core of the median's hardness,
{\sc MedianBit}, and its complement do not have non-trivial lower bounds;
hence current time-space tradeoff lower bound techniques are powerless for 
computing the median.

We conjecture that the lower bound $T=\Omega(n\log\log_S n)$ holds
for finding the median using general non-oblivious algorithms as well as
oblivious and comparison algorithms.

\ignore{
A stronger lower bound technique is necessary to prove this conjecture.   
In particular, the technique for general algorithms just described treats the
decomposition
of $f_B$ into $\bigvee_\tau F_\tau$ by analyzing each $F_\tau$ in isolation,
independent of the fact that the traces partition the input space.
The communication complexity technique that we used for the oblivious case
implicitly takes advantage of this partitioning; it must do so not only
because of the nondeterministic algorithm just described, but also because the 
the lower bound only holds if the number of rounds is limited.
Beyond using the fact that the traces partition the input, another aspect
that the standard proof method does not take advantage of is the fact that the
next node in the trace on a given input must be computed ``on
the fly" after having seen relatively few input values.
}

%\newpage
\section*{Acknowledgements}

We thank Anup Rao for suggesting the improved form of Lemma~\ref{lem:rel} that
we include here.

{%\small
\renewcommand{\baselinestretch}{0.96}
\bibliographystyle{plain} %{plain}
\bibliography{theory}

\def\bibRCS{$Id: theory.bib,v 1.59 2015/03/25 08:25:16 beame Exp beame $}
  \makeatletter \@ifundefined{ccisdefined}{ \newcommand{\cc}[1]{\mbox{\it
  #1\/}} \newcommand{\ccisdefined}{} }{} \@ifundefined{journalfont}{
  \newcommand{\journalfont}{\it } }{} \makeatother
\begin{thebibliography}{10}

\bibitem{ajtai:nonlinear}
M.~Ajtai.
\newblock A non-linear time lower bound for boolean branching programs.
\newblock In {\em Proceedings 40th Annual Symposium on Foundations of Computer
  Science}, pages 60--70, New York,NY, October 1999. IEEE.

\bibitem{ajtai:nondetBP-journal}
M.~Ajtai.
\newblock Determinism versus non-determinism for linear time {RAMs} with memory
  restrictions.
\newblock {\em Journal of Computer and System Sciences}, 65(1):2--37, August
  2002.

\bibitem{am:meanders}
Noga Alon and Wolfgang Maass.
\newblock Meanders and their applications in lower bounds arguments.
\newblock {\em Journal of Computer and System Sciences}, 37:118--129, 1988.

\bibitem{bea:sorting}
P.~Beame.
\newblock A general sequential time-space tradeoff for finding unique elements.
\newblock {\em SIAM Journal on Computing}, 20(2):270--277, 1991.

\bibitem{bcm:windowsED}
P.~Beame, R.~Clifford, and W.~Machmouchi.
\newblock Element distinctness, frequency moments, and sliding windows.
\newblock In {\em Proceedings of the 54th Annual Symposium on Foundations of
  Computer Science}, pages 290--299, Berkeley, CA, October 2013. IEEE.

\bibitem{bssv:randomts-journal}
P.~Beame, M.~Saks, X.~Sun, and E.~Vee.
\newblock Time-space trade-off lower bounds for randomized computation of
  decision problems.
\newblock {\em Journal of the ACM}, 50(2):154--195, 2003.

\bibitem{bjs:tradeoff}
Paul Beame, T.~S. Jayram, and Michael Saks.
\newblock Time-space tradeoffs for branching programs.
\newblock {\em Journal of Computer and System Sciences}, 63(4):542--572,
  December 2001.

\bibitem{bv:multibp}
Paul Beame and Erik Vee.
\newblock Time-space tradeoffs, multiparty communication complexity, and
  nearest-neighbor problems.
\newblock In {\em Proceedings of the Thirty-Fourth Annual ACM Symposium on
  Theory of Computing}, pages 688--697, Montreal, Quebec, Canada, May 2002.

\bibitem{bfprt:median}
Manuel Blum, Robert~W. Floyd, Vaughan~R. Pratt, R.~L. Rivest, and Robert~E.
  Tarjan.
\newblock Time bounds for selection.
\newblock {\em Journal of Computer and System Sciences}, 7(4):448--461, 1972.

\bibitem{bc:sorting}
Allan Borodin and Stephen~A. Cook.
\newblock A time-space tradeoff for sorting on a general sequential model of
  computation.
\newblock {\em SIAM Journal on Computing}, 11(2):287--297, May 1982.

\bibitem{brs93}
Allan Borodin, A.~A. Razborov, and Roman Smolensky.
\newblock On lower bounds for read-$k$ times branching programs.
\newblock {\em Computational Complexity}, 3:1--18, October 1993.

\bibitem{cjp:median}
A.~Chakrabarti, T.~S. Jayram, and M.~Patrascu.
\newblock Tight lower bounds for selection in randomly ordered streams.
\newblock In {\em Proceedings of the Nineteenth Annual {ACM-SIAM} Symposium on
  Discrete Algorithms}, pages 720--729, San Francisco, CA, January 2008.
  Society for Industrial and Applied Mathematics.

\bibitem{chan:selection-journal}
T.~M. Chan.
\newblock Comparison-based time-space lower bounds for selection.
\newblock {\em ACM Transactions on Algorithms}, 6(2):26:1--16, 2010.

\bibitem{frederickson:sortselect-journal}
Greg~N. Frederickson.
\newblock Upper bounds for time-space trade-offs in sorting and selection.
\newblock {\em Journal of Computer and System Sciences}, 34(1):19--26, 1987.

\bibitem{holenstein:parallel-journal}
T.~Holenstein.
\newblock Parallel repetition: Simplification and the no-signaling case.
\newblock {\em Theory of Computing}, 5(1):141--172, 2009.

\bibitem{jrsw:pvnpconp}
S.~Jukna, A.~A. Razborov, P.~Savick{\'y}, and I.~Wegener.
\newblock On {P} versus {NP$\cap$ co-NP} for decision trees and read-once
  branching programs.
\newblock {\em Computational Complexity}, 8(4):357--370, 1999.

\bibitem{kn97}
E.~Kushilevitz and N.~Nisan.
\newblock {\em Communication Complexity}.
\newblock Cambridge University Press, Cambridge, England ; New York, 1997.

\bibitem{mp:selectsort-journal}
J.~Ian Munro and Michael~S. Paterson.
\newblock Selection and sorting with limited storage.
\newblock {\em Theoretical Computer Science}, 12:315--323, 1980.

\bibitem{mr:selectrom-journal}
J.~Ian Munro and Venkatesh Raman.
\newblock Selection from read-only memory and sorting with minimum data
  movement.
\newblock {\em Theoretical Computer Science}, 165(2):311--323, 1996.

\bibitem{okol93}
E.~Okol'nishnikova.
\newblock On lower bounds for branching programs.
\newblock {\em Siberian Advances in Mathematics}, 3(1):152--166, 1993.

\bibitem{pr:comparison-sorting}
J.~Pagter and T.~Rauhe.
\newblock Optimal time-space trade-offs for sorting.
\newblock In {\em Proceedings 39th Annual Symposium on Foundations of Computer
  Science}, pages 264--268, Palo Alto, CA, November 1998. IEEE.

\bibitem{yao:ED-journal}
A.~C.-C. Yao.
\newblock Near-optimal time-space tradeoff for element distinctness.
\newblock {\em SIAM Journal on Computing}, 23(5):966--975, 1994.

\end{thebibliography}
}
\newpage
\appendix
\section{Proof of Lemma~\ref{lem:rel}}
The proof is inspired by that of the parallel repetition theorem.
For $(x,y)$ chosen according to $\calD$, we first design a public coin
protocol in which the players randomly choose $i\in [k]$, and
with small probability of error jointly choose a
random vector $W^{-i}$ of values consisting of exactly one of $x_j$ or $y_j$ for
each $j\ne i$ and a random message $M$ for Alice consistent with those inputs
whose distribution is close to that of Alice's first message.
The players then independently use the public coins to randomly complete their
inputs to be consistent with $\calD$ on each coordinate (and in Alice's
case consistent with the message agreed upon).
The resulting protocol will have expected error at most $\epsilon+\delta$.
We then fix the public coins (and hence all inputs other than $(x,y)$)
to create the claimed deterministic protocol.

Let $X_i$ and $Y_i$ for $i\in [k]$ denote the random variables associated with
the components of the distribution $\calD^k$.
Let $M$ denote the random variable for Alice's first message.
Define the random variable $W_i$ that is $X_i$ with probability 1/2 and $Y_i$
with probability 1/2.   Let $W$ denote the random variable $W_1\ldots W_k$
and $W^{-i}$ denote the variable $W$ with $W_i$ removed.
Then 
\begin{align*}
m_1&\ge H(M)\\
&\ge I(M; X_1 Y_1 \ldots X_k Y_k | W)\qquad\mbox{by definition}\\
&\ge \sum_{i=1}^k I(M; X_i Y_i | W)\\
\noalign{\mbox{since the $X_iY_i$ are conditionally independent given $W$}}
&= \sum_{i=1}^k I(M; X_i Y_i | W_i W^{-i})\\
&= \sum_{i=1}^k \frac{I(M; X_i Y_i | X_i W^{-i})+I(M; X_i Y_i | Y_i W^{-i})}{2}\qquad\mbox{by definition of }W_i\\
&= \sum_{i=1}^k \frac{I(M; Y_i | X_i W^{-i})+I(M; X_i | Y_i W^{-i})}{2}\\
&= \sum_{i=1}^k \frac{I(MW^{-i}; Y_i | X_i)-I(W^{-i}; Y_i | X_i) +I(MW^{-i}; X_i | Y_i )-I(W^{-i}; X_i | Y_i)}{2}\\
\noalign{\mbox{by the chain rule}}
&= \sum_{i=1}^k \frac{I(MW^{-i}; Y_i | X_i)+I(MW^{-i}; X_i | Y_i )}{2}\\
\noalign{\mbox{since $W^{-i}$ is independent of $X_iY_i$.}}
\end{align*}
\begin{sloppypar}
Since $m_1\le \delta^2 k/(8\ln 2)$, it follows that
$\E_{i\in [k]} (I(MW^{-i}; Y_i | X_i)+I(MW^{-i}; X_i | Y_i ))\le \delta^2/(4\ln 2)$.
We use this to derive that in expectation over random choices of $i$, and
$(x,y)$ chosen from $\calD$, 
the distributions $MW^{-i}|X_i{=}x$ and $MW^{-i}|Y_i{=}y$ are both
statistically close to the distribution $MW^{-i}|X_i{=}x,Y_i{=}y$.
We now use the following proposition which follows from Pinsker's inequality.
\end{sloppypar}

\begin{proposition}
\label{prop:pinsker}
Let $P$ and $Q$ be probability distributions.  \\
Then $\E_{q\sim Q} ||P-(P|Q{=}q)||^2\le \frac{\ln 2}{2} I(P;Q)$. 
\end{proposition}

It follows that
\begin{align*}
&\E_{(x,y)\sim \calD} ||(MW^{-i}|X_i{=}x)-(MW^{-i}|X_i{=}x,Y_i{=}y)||^2\\
&= \E_{x} \E_{y: (x,y)\sim \calD|X=x} ||(MW^{-i}|X_i{=}x)-(MW^{-i}|X_i{=}x,Y_i{=}y)||^2\\
&\le \frac{\ln 2}{2} \E_{x} I(MW^{-i};Y_i|X_i{=}x) \qquad\mbox{by Proposition~\ref{prop:pinsker}}\\
&=\frac{\ln 2}{2} I(MW^{-i};Y_i|X_i)\qquad\mbox{by definition}.
\end{align*}
The analogous bound shows that the expected square of the statistical distance
between the distributions $MW^{-i}|Y_i{=}y$ and $MW^{-i}|X_i{=}x,Y_i{=}y$ is
at most $\frac{\ln 2}{2} I(MW^{-i};X_i|Y_i)$.

\begin{sloppypar}
Therefore in expectation over choices of $(x,y)$ and $i$, the sum of the
squares of the statistical distances between
the distributions $MW^{-i}|X_i{=}x$ and $MW^{-i}|X_i{=}x,Y_i{=}y$ 
and between $MW^{-i}|Y_i{=}y$ and $MW^{-i}|X_i{=}x,Y_i{=}y$ is at most $\delta^2/8$.
Write $\epsilon_{x,y,i,X}= ||(MW^{-i}|X_i{=}x)-(MW^{-i}|X_i{=}x,Y_i{=}y)||$
and
$\epsilon_{x,y,i,Y}= ||(MW^{-i}|Y_i{=}y)-(MW^{-i}|X_i{=}x,Y_i{=}y)||$.
In particular we have
\begin{align*}
\E_{(x,y)\sim \calD}&\E_{i\in [k]} (\epsilon^2_{x,y,i,X}+\epsilon^2_{x,y,i,Y})\\
&\le \frac{\ln 2}{2}
\E_{(x,y)\sim \calD}\E_{i\in [k]} 
(I(MY^{-1};X_i|Y_i)+I(MW^{-i};Y_i|X_i))\\
&\le \delta^2/8.
\end{align*}
and hence
$$\E_{(x,y)\sim \calD}\E_{i\in [k]} (\epsilon_{x,y,i,X}+\epsilon_{x,y,i,Y})
\le \delta/2.$$
\end{sloppypar}

We now apply Holenstein's Lemma~\cite{holenstein:parallel-journal} to say
that given $x,y$ and $i$, except for a failure probability of at most
$2(\epsilon_{x,y,i,X}+\epsilon_{x,y,i,Y})$, 
Alice and Bob without any communication 
can use the shared random string to agree on a sample $(m,w^{-i})$ from 
$MW^{-i}|X_i{=}x,Y_i{=}y$.
Therefore the expected failure probability is at most $\delta$.

Once Alice and Bob have selected $(m,w^{-i})$, Alice uses private randomness
to sample the remainder of her input from
$X_1\ldots X_k| M{=}m,W^{-i}{=}w^{-i},X_i{=}x$.
Bob independently uses private randomness
to sample the remainder of his input from
\begin{align*}
&Y_1\ldots Y_k| M{=}m,W^{-i}{=}w^{-i},Y_i{=}y 
\\&= Y_1\ldots Y_k| W^{-i}{=}w^{-i},Y_i{=}y
\end{align*}
which only depends on $\calD^k$.

Then Alice and Bob simulate the remainder of the protocol starting with the
second message overall; i.e., Bob's
first message.  The difference in the distribution from $\calD^{[k]}$ on
the result is at most $\delta$ so the expected error is
at most $\epsilon+\delta$.

\end{document}